\documentclass[11pt]{article}

\usepackage{amsfonts,amsmath,amsthm,amssymb,bbm}

\usepackage{fullpage,times,array}
\newcolumntype{T}{>{\tiny}l} 
\newcolumntype{H}{>{\Huge}l} 

\usepackage{longtable}

\usepackage{hyperref}
\hypersetup{colorlinks=false}
\usepackage[shortlabels]{enumitem}

\usepackage[dvipsnames]{xcolor}
\usepackage{graphicx}

\usepackage[final]{listlbls}

\usepackage[nothing]{algorithm}  
\usepackage[noend]{algpseudocode}

\newtheorem{theorem}{Theorem} 
\newtheorem{lemma}{Lemma}
\newtheorem{definition}{Definition}

\newtheorem{remark}{Remark}

\def\ex{\mathbb E}

\title{Probabilistic Model Counting with Short XORs}

\author{
Dimitris Achlioptas
\thanks{Research supported by NSF grant CCF-1514128 and grants from Adobe and Yahoo!} 
\\ Deparment of Computer Science\\ University of California Santa Cruz
\and 
Panos Theodoropoulos\thanks{Research supported by the Greek State Scholarships Foundation (IKY).} 
 \\ Deparment of Informatics \& Telecommunications \\ University of Athens
}

\date{\empty}

\begin{document}

\maketitle

\begin{abstract}
The idea of counting the number of satisfying truth assignments (models) of a formula by adding random parity constraints can be traced back to the seminal work of Valiant and Vazirani, showing that NP is as easy as detecting unique solutions. While theoretically sound, the random parity constraints in that construction have the following drawback: each constraint, on average, involves half of all variables. As a result, the branching factor associated with searching for models that also satisfy the parity constraints  quickly gets out of hand. In this work we prove that one can work with much shorter parity constraints and still get rigorous mathematical guarantees, especially when the number of models is large so that many constraints need to be added. Our work is based on the realization that the essential feature for random systems of parity constraints to be useful in probabilistic model counting is that the geometry of their set of solutions resembles an error-correcting code.
\end{abstract}
 
\section{Introduction}

Imagine a blind speaker entering an amphitheater, wishing to estimate the number of people present. She starts by asking ``Is anyone here?" and hears several voices saying ``Yes." She then says ``Flip a coin inside your head; if it comes up heads, please never answer again." She then asks again  ``Is anyone here?" and roughly half of the people present say ``Yes." She then asks them to flip another coin, and so on. When silence occurs after $i$ rounds, she estimates that approximately $2^i$ people are present.

Given a CNF formula $F$ with $n$ variables we would like to approximate the size of its set of satisfying assignments (models), $S=S(F)$, using a similar approach. Following the path pioneered by~\cite{valiant1986np}, \cite{stockmeyer1985approximation}, and \cite{sipser1983complexity}, in order to check if $|S(F)| \approx 2^i$ we  form a random set $R_i \subseteq \{0,1\}^n$ such that $\Pr[\sigma \in R_i] = 2^{-i}$ for every $\sigma \in \{0,1\}^n$ and determine if $|S(F) \cap R_i| \approx 1$. The key point is to represent $R_i$ \emph{implicitly} as the set of solutions to a system of $i$ random linear equations modulo 2 (parity constraints). Thus, to determine $|S(F) \cap R_i|$ we simply add the parity constraints to $F$ and invoke CryptoMiniSAT~\cite{soos2009cryptominisat}, asking it to determine the number of solutions of the combined set of constraints (clauses and parity equations). 

There has already been a long line of practical work along the lines described above for model counting, including~\cite{mbound,gomes2006near,ghss07:shortxors,wishicml13,ermon2014low} and~\cite{DBLP:conf/cp/ChakrabortyMV13,DBLP:conf/cav/ChakrabortyMV13,DBLP:conf/dac/ChakrabortyMV14,DBLP:conf/aaai/ChakrabortyFMSV14,DBLP:conf/ijcai/ChakrabortyMV16,DBLP:journals/constraints/IvriiMMV16}. In most of these works, each parity constraint includes each variable independently and with probability 1/2, so that each parity constraint includes, on average, $n/2$ variables. While systems of such long parity constraints have the benefit that membership in their set of solutions enjoys pairwise independence, making the probabilistic analysis very simple, the length of the constraints can be a severe limitation. This fact was recognized at least as early as~\cite{ghss07:shortxors}, and efforts have been made in some of the aforementioned works to remedy it, by considering parity equations where each constraint still includes each variable independently, but with probability $p<1/2$. While such sparsity helps with computing $|S(F) \cap R|$, the statistical properties of the resulting random sets, in particular the variance of $|S(F) \cap R|$, deteriorate rapidly as $p$ decreases~\cite{ermon2014low}.

In this work we make two contributions. First, we show that bounding $|S(F)|$ from below and from above should be thought of as two separate problems, the former being much easier than the latter. Secondly, we propose the use of random systems of parity equations corresponding to the parity-check matrices of low density binary \emph{error-correcting codes}. These matrices are also sparse, but their entries are \emph{not indepedent}, causing them to have statistical properties dramatically better than those of similarly sparse i.i.d.\ matrices. As a result, they can be used to derive \emph{upper} bounds on $|S(F)|$, especially when $\log_2|S(F)| = \Omega(n)$.
\section{Lower Bounds are Easy}

For a distribution $\mathcal{D}$, let $R \sim \mathcal{D}$ denote that random variable $R$ is distributed according to $\mathcal{D}$. 

\begin{definition} 
Let $\mathcal{D}$ be a distribution on subsets of a set $U$ and let $R \sim \mathcal{D}$. We say that $\mathcal{D}$ is $i$-uniform if $\Pr[\sigma \in R] = 2^{-i}$ for every $\sigma \in U$. 
\end{definition}

When $U=\{0,1\}^n$, some examples of $i$-uniform distributions  are:
\begin{enumerate}[(i)]
\item
$R$ contains each  $\sigma \in \{0,1\}^n$ independently with probability $2^{-i}$. 
\item
$R$ is a uniformly random subcube of dimension $n-i$.
\item
$R = \{\sigma: A \sigma = b\}$, where $A \in \{0,1\}^{i \times n}$ is \emph{arbitrary} and $b \in \{0,1\}^i$ is \emph{uniformly} random.
\end{enumerate}

\emph{Any} $i$-uniform distribution $\mathcal{D}_i$ can be used to compute a \emph{rigorous} lower bound on the number of satisfying truth assignments of a formula (models), as follows. (We note that CryptoMiniSAT has an optional cutoff parameter $s \ge 1$ such that as soon as $s$ solutions are found, the search stops without completing. We use this capacity in line~\ref{lin:clever} of Algorithm~\ref{alg:lower} and lines~\ref{not_too_small}, \ref{true_work} of Algorithm~\ref{alg:simple}.)
\begin{algorithm}
\caption{Decides if $|S| \ge 2^i$ with 1-sided error probability $\theta > 0$}\label{alg:lower}
\begin{algorithmic}[1]
\State $t \leftarrow \lceil 8 \ln(1/\theta) \rceil$ \label{lin:iter} 
\Comment{{\small $\theta$ is the desired error probability bound}}
\State $Z \leftarrow 0$
\State $j \leftarrow 0$
	\While{$j < t$ and $Z < 2t$}\label{lin:extra} \Comment{{\small The condition $Z < 2t$ is an optimization}}
		\State Sample $R_{j} \sim \mathcal{D}_i$ \Comment{{\small Where $\mathcal{D}_i$ is any $i$-uniform distribution}}				
		\State $Y_{j} \leftarrow \min\{4,|S(F) \cap R_{j}|\}$ \label{lin:clever}	
			\Comment{{\small  Seek up to 4 elements of $S(F) \cap R_{j}$}} 		
		\State $Z \leftarrow Z + Y_j$					
		\State $j \leftarrow j+1$
	\EndWhile											
	\If{$Z/t \ge 2$} 									\label{lin:why2}	
	\State \Return ``I believe that $|S(F)| \ge 2^i$"		
	\Else
	\State \Return ``Don't know"
\EndIf
\end{algorithmic}
\end{algorithm}

The only tool we use to analyze Algorithm~\ref{alg:lower} is Hoeffding's Inequality. 
\begin{lemma}[Hoeffding's Inequality] If $Z = Y_1 + \cdots + Y_t$, where $0 \le Y_i \le b$ are independent random variables, then for any $w \ge 0$,
\[
\Pr[Z/t \ge \ex Z/t + w] \le 
\exp\left(-2t\left(\frac{w}{b}\right)^2\right) 
\quad
\text{and}
\quad
\Pr[Z/t \le \ex Z/t - w] \le 
\exp\left(-2t\left(\frac{w}{b}\right)^2\right) 
\enspace .
\]
\end{lemma}

\begin{theorem}\label{thm:lower_b}
The output of Algorithm~\ref{alg:lower} is incorrect with probability at most $\theta$.
\end{theorem}
\begin{proof}
Let $S=S(F)$. For the algorithm's output to be incorrect it must be that $|S| < 2^i$ and $Z/t \ge 2$. If $|S| < 2^i$, then $\ex Y_j \le |S|2^{-i} < 1$, implying $\ex Z/t < 1$. Since $Z$  is the sum of $t$ i.i.d. random variables $0 \le Y_{j} \le 4$, Hoeffding's inequality implies that $\Pr[Z/t \ge 2] \le \Pr[Z/t \ge \ex Z/t + 1]
\le 
\exp
\left(
-t/8
\right) $.
\end{proof}

Notably, Theorem~\ref{thm:lower_b} does not address the \emph{efficacy} of Algorithm~\ref{alg:lower}, i.e., the probability of ``Yes" when $|S| \ge 2^i$. As we will see, bounding this probability from below requires much more than mere $i$-uniformity.

\subsection{Choice of constants}

We can make Algorithm~\ref{alg:lower} more likely to return ``Yes" instead of ``Don't know" by increasing the number 4 in line~\ref{lin:clever} and/or decreasing the number 2 in lines~\ref{lin:extra},\ref{lin:why2}. Each such change, though, increases the number of iterations, $t$, needed to achieve the same probability of an erroneous ``Yes". The numbers 4, 2 appear to be a good balance in practice between the algorithm being fast and being useful.

\subsection{Dealing with Timeouts}

Line~\ref{lin:clever} of Algorithm~\ref{alg:lower} requires determining $\min\{4,|S(F) \cap R_{j}|\}$. Timeouts may prevent this from happening, since  in the allotted time the search may only find $s < 4$ elements of $S(F) \cap R_{j}$ but not conclude that no other such elements exist. Nevertheless, if $Y_j$ is always set to a number \emph{no greater} than $\min\{4,|S(F) \cap R_{j}|\}$, then both Theorem~\ref{thm:lower_b} and its proof remain valid. So, for example, whenever timeout occurs, we can set $Y_j$ to the number $s <4$ of elements of $S(F) \cap R_{j}$ found so far. Naturally, the modification may increase the probability of ``Don't know", e.g., if we trivially always set $Y_j \leftarrow 0$.

\subsection{Searching for a Lower Bound}\label{sec:bis_search}

We can derive a lower bound for $\log_2|S(F)|$ by invoking Algorithm~\ref{alg:lower} with $i=1,2,3,\ldots,n$ sequentially. Let $\ell$ be the greatest integer for which the algorithm returns ``Yes" (if any). By Theorem~\ref{thm:lower_b}, $\log_2|S(F)|\ge \ell$ with probability at least $1-\theta n$, as the probability that at least one ``Yes" answer is wrong is at most $\theta n$.

There is no reason, though, to increase $i$ sequentially. We can be more aggressive and invoke Algorithm~\ref{alg:lower} with $i=1,2,4,8,\ldots$ until we encounter our first ``Don't know", say at $i = 2^{u}$. At that point we can perform binary search in $ \{2^{u-1}, \ldots , 2^{u}-1\}$, treating every ``Don't know" answer as a (conservative) imperative to reduce the interval's upper bound to the midpoint and every ``Yes" answer as an allowance to increase the interval's lower bound to the midpoint. We call this scheme ``doubling-binary search." In fact, even this scheme can be accelerated by running Algorithm~\ref{alg:lower} with the full number of iterations only for values of $i$ for which we have good evidence of being lower bounds for $\log_2|S(F)|$. Specifically, Algorithm~\ref{alg:augment} below takes as input an arbitrary lower bound $0 \le \ell \le \log_2  |S|$ and tries to improve it. 

\begin{theorem}\label{thm:augment}
The output of Algorithm~\ref{alg:augment} exceeds $\log_2|S|$ with probability at most $\theta$. 
\end{theorem}
\begin{proof}
For the answer to be wrong it must be that some invocation of Algorithm~\ref{alg:lower} in line~\ref{lin:msd} with $i > \log_2|S|$ returned ``Yes". Since Algorithm~\ref{alg:augment} invokes Algorithm~\ref{alg:lower} in line~\ref{lin:msd} at most $\lceil \log_2 n\rceil$ times, and in each such invocation we set the failure probability to $\theta / \lceil \log_2 n\rceil$, the claim follows.
\end{proof}

Let $A1(i,z)$ denote the output of Algorithm~\ref{alg:lower} if line~\ref{lin:iter} is replaced by $t \leftarrow z$.
\begin{algorithm}
\caption{Given $0 \le \ell \le \log_2|S(F)|$ returns $i \ge \ell $ such that 
$|S| \ge 2^i$ with probability at least $1-\theta$}\label{alg:augment}
\begin{algorithmic}[1]
\State $i \leftarrow \ell$ \Comment{$0 \le \ell \le \log_2|S(F)|$}
\State
\State $j \leftarrow 0$\Comment{{\small Doubling search until first ``Don't Know"}}
\While{\{$A1(\ell + 
2^j,1)$ = ``Yes"\}} 
	\State $j \leftarrow j+1$ 
\EndWhile
\If {$j=0$}
\Return $\ell$ \Comment{{\small Failed to improve upon $\ell$}}
\EndIf
\State
\State $h\leftarrow \ell + 
2^{j} - 1$  	\Comment{{\small{First ``Don't Know" occurred at $h+1$}}}
\State $i \leftarrow \ell + 
2^{j-1}$	\Comment{{\small{Last ``Yes" occurred at $i$}}}
\State 
\While {$i < h$} \Comment{{\small{Binary search for ``Yes" in $[i,h]$}}}
	\State $m \leftarrow i + \lceil (h-i)/2\rceil$ 
	\Comment{{\small{where $i$ holds the greatest seen ``Yes" and}}}
     \If {$A1(m,1)$ = ``Yes"}
     \Comment{{\small{$h+1$ holds the smallest seen ``Don't know"}}}
    		\State $i \leftarrow m$  
	\Else
		\State $h \leftarrow m-1$
	\EndIf
\EndWhile
\State \Comment{{\small{$i$ holds the greatest seen ``Yes"}}}
\State $j \leftarrow 1$
\Repeat 		
	\Comment{{\small{Doubling search backwards starting}}}
	\State $i \leftarrow i-2^j$
	\Comment{{\small{ from $i-2$ for a lower bound that}}}
	\State $j \leftarrow j+1$ 
	\Comment{{\small{holds with probability $1-\theta/\lceil\log_2 n\rceil$}}}
\Until{$i \le \ell
$ {\bf{or}} $A1(i, \lceil 8 \ln(\lceil\log_2 n\rceil/\theta) \rceil)$ = ``Yes"} \label{lin:msd}
\State \Return $\max\{\ell, i\}$
\end{algorithmic}
\end{algorithm}

\section{What Does it Take to Get a \emph{Good} Lower Bound?}

The greatest $i \in [n]$ for which Algorithms~\ref{alg:lower}, \ref{alg:augment} will return ``Yes" may be arbitrarily smaller than $\log_2 |S|$. The reason for this, at a high level, is that even though $i$-uniformity is enough for the \emph{expected} size of $S \cap R$ to be $2^{-i}|S|$, the actual size of $S \cap R$ may behave like the winnings of a lottery: typically zero, but huge with very small probability. So, if we add $j = \log_2|S|-\Theta(1)$ constraints, if a lottery phenomenon is present, then even though the expected size of $S \cap R$ is greater than 1, in any realistic number of trials we will always see $S \cap R = \emptyset$, in exactly the same manner that anyone playing the lottery a realistic number of times will, most likely, never experience winning. Yet, concluding at that point that $|S| < 2^j$ is obviously wrong.

The above discussion makes it clear that the heart of the matter is controlling the \emph{variance} of $|S \cap R|$. One way to do this is by bounding the ``lumpiness" of the sets in the support of the distribution $\mathcal{D}_i$, as measured by the quantity defined below, which measures  lumpiness at a scale of $M$ (the smaller the quantity in~\eqref{eq:boost_def}, the less lumpy the distribution, the smaller the variance).
\begin{definition} 
Let $\mathcal{D}$ be any distribution on subsets of $\{0,1\}^n$ and let $R \sim \mathcal{D}$. For any fixed $M \ge 1$,
\begin{equation}\label{eq:boost_def}
\mathrm{Boost}(\mathcal{D}, M) = \max_{\substack{S \subseteq \{0,1\}^n\\ |S| \ge M}}
\frac{1}{|S|(|S|-1)}
\sum_{\substack{\sigma, \tau \in S\\ \sigma \neq \tau}}\frac{\Pr[\sigma, \tau \in R]}{\Pr[\sigma \in R] \Pr[\tau \in R]}\enspace . 
\end{equation}
\end{definition}
To get some intuition for~\eqref{eq:boost_def} observe that the ratio inside the sum equals the factor by which the a priori probability that a truth assignment belongs in $R$ is modified by conditioning on some other truth assignment belonging in $R$. For example, if membership in $R$ is pairwise independent, then $\mathrm{Boost}(\mathcal{D},\cdot)=1$, i.e., there is no modification. Another thing to note is that since we require $|S| \ge M$ instead of $|S|=M$ in~\eqref{eq:boost_def}, the function $\mathrm{Boost}(\mathcal{D},\cdot)$ is non-increasing.

While pairwise independence is excellent from a statistical point of view, the corresponding geometry of the random sets $R$ tend to be make determining $|S \cap R|$ difficult, especially when $|R|$ is far from the extreme values $2^n$ and $1$. And while there are even better distributions statistically, i.e., distributions for which $\mathrm{Boost}(\mathcal{D},\cdot) < 1$, those are even harder to work with. In the rest of the paper we restrict to the case $\mathrm{Boost}(\mathcal{D}, \cdot) \ge 1$ (hence the name $\mathrm{Boost}$). As we will see, the crucial requirement for an $i$-uniform distribution $\mathcal{D}_i$ to be useful is that $\mathrm{Boost}(\mathcal{D}_i, \Theta(2^i))$ is relatively small, i.e., an $i$-uniform distribution can be useful even if $\mathrm{Boost}(\mathcal{D}_i)$ is large for sets of size much less than $2^i$. 
The three examples of $i$-uniform distributions discussed earlier differ dramatically in terms of $\mathrm{Boost}$.
\begin{enumerate}[(i)]
\item
When $R$ is formed by selecting each element of $\{0,1\}^n$ independently, trivially, $\mathrm{Boost}(\mathcal{D},\cdot) =1$ since we have full, not just pairwise, independence. But just specifying $R$ in this case requires space proportional to $|R|$.
\item
When $R$ is a random subcube, both specifying $R$ is extremely compact and searching for models in $R$ is about as easy as one could hope for: just ``freeze" a random subset of $i$ variables and search over the rest. Unfortunately, random subcubes are terrible statistically, for example exhibiting huge $\mathrm{Boost}$ when $S$ itself is a subcube. 
\item
When $R = \{\sigma : A \sigma = b\}$, where $b \in \{0,1\}^i$ is uniformly random, the distribution of $A$ is crucial.
\begin{itemize}
\item 
At one end of the spectrum, if $A = \mathbf{0}$, then $R$ is either the empty set or the entire cube $\{0,1\}^n$, depending on whether $b = \mathbf{0}$ or not. Thus, $\mathrm{Boost}(\mathcal{D},\cdot) = 2^{i}$, the maximum possible.
\item
The other end of the spectrum occurs when $A$ is uniform in $\{0,1\}^{i \times n}$, i.e., when the entries of $A$ are independently 0 or 1 with equal probability. In this case $\mathrm{Boost}(\mathcal{D},\cdot) =1$, a remarkable and well-known fact that can be seen by the following simple argument. By the definition of $R$,
\[
\Pr[\sigma, \tau \in R] = \Pr[A \sigma = A \tau \wedge A \sigma = b] = \Pr[A(\sigma-\tau) = 0] \cdot \Pr[A \sigma = b] \enspace ,
\] 
where to see the second equality imagine first selecting the matrix $A$ and only then the vector $b$. To prove that 
$\Pr[A(\sigma-\tau) = 0] = 2^{-i}$, and thus conclude the proof, select any non-zero coordinate of $\sigma - \tau$, say $j$ (as $\sigma \neq \tau$ there is at least one). Select arbitrarily all entries of $A$, except those in the $j$-th column. Observe now that whatever the selection is, there is exactly one choice for each entry in the $j$-th column such that $A(\sigma - \tau) = 0$. Since the $i$ elements of the $j$-th column are selected uniformly and independently, the claim follows.
\end{itemize}
\end{enumerate} 
\section{Accurate Counting Without Pairwise Independence}

For each $i \in \{0,1,\ldots,n\}$ let $\mathcal{D}_i$ be some (arbitrary) $i$-uniform distribution on subsets of $\{0,1\}^n$. (As mentioned, we consider $F$, and thus $n$, fixed, so that we can write $\mathcal{D}_{i}$ instead of $\mathcal{D}_{i,n}$ to simplify notation.) We will show that given any $0 \le L \le |S|$, we can get a  \emph{guaranteed} probabilistic approximation of $|S|$ in time proportional to the square of $B = \max_{\ell \le i \le n} \mathrm{Boost}(\mathcal{D}_i, 2^i)$. We discuss specific distributions and approaches to bounding $B$ in Sections~\ref{sec:ldpc} and~\ref{sec:moufa}. 

\begin{algorithm}
\caption{Given $0 \le L \le |S|$, returns a number in $(1\pm \delta) |S|$ with probability at least $1-\theta$}\label{alg:simple}
\begin{algorithmic}[1]
\State Choose the approximation factor $\delta \in (0,1/3]$  
		\Comment{{\small See Remark~\ref{more_tune}}}
\State Choose the failure probability $\theta>0$ 
\State
\State Receive as input any number $0 \le L \le |S|$ \Comment{{\small We can set $L=0$ by default}}
\State
\If {$|S| < 4/\delta$} \Return $|S|$
		\Comment{{\small If $L < 4/\delta$ then seek up to $4/\delta$ models of $F$}}\label{not_too_small}

\EndIf
\State

\State Let $\ell = \max\{0,\lfloor \log_2(\delta L/4) \rfloor\}$   \label{lin:lll}
		
\State
\If {all distributions $\{\mathcal{D}_i\}_{i=\ell}^n$ enjoy pairwise independence} 		
\State $B \leftarrow 1$
\Else
\State Find $\displaystyle{B \ge \max_{\ell \le i \le n}\mathrm{Boost}(\mathcal{D}_i, 2^i)}$
		\Comment{{\small See Sections~\ref{sec:ldpc},\ref{sec:moufa}}}
\EndIf
\State
\State $\xi \leftarrow 8/\delta$ 
\State $b \leftarrow \lceil\xi + 2( \xi + \xi^2(B-1)) \rceil$ \label{bline}
		\Comment{{\small If $B=1$, then $b=\lceil 24/\delta \rceil$}}
\State $ t \leftarrow \lceil (2b^2/9) \ln(2n/\theta)\rceil$ \label{tline}
		
\State
\For{$i$ from $\ell$ to $n$} \label{main_loop_start}
	\State $Z_i \leftarrow 0$
	\For{$j$ from 1 to $t$}\label{inner_loop}
	\State Sample $R_{i,j} \sim \mathcal{D}_i$
	\State Determine $Y_{i,j} = \min\{b,|S \cap R_{i,j}|\}$ \label{true_work} 
	\Comment{{\small Seek up to $b$ elements of $S(F) \cap R_{i,j}$}} 
	\State $Z_i \leftarrow Z_i + Y_{i,j}$
	\EndFor
	\State $A_i \leftarrow Z_i/t$
\EndFor
\State $j \leftarrow \max\{i \ge \ell:  A_i \ge (1-\delta) (4/\delta)\}$ \label{j_choice}
\State \Return $\max\{L,A_j 2^j\}$
\end{algorithmic}
\end{algorithm}

\begin{theorem}\label{thm:main}
The output of Algorithm~\ref{alg:simple} lies outside the range $(1\pm \delta)|S|$ with probability at most $\theta$.
\end{theorem}

\begin{remark}\label{more_tune}
Algorithm~\ref{alg:simple} can be modified to have accuracy parameters $\beta < 1$ and $\gamma > 1$ so that its output is between $\beta |S|$ and $\gamma|S|$ with probability at least $1-\theta$. At that level of generality, both the choice of $\xi$ and the criterion for choosing $j$ in line~\ref{j_choice} must be adapted to $\beta,\gamma,\theta$. Here we focus on the high-accuracy case, choosing $\xi,b,t$ with simple form.
\end{remark}

Let $q = \lfloor \log_2(\delta|S|/4)\rfloor$. We can assume that $q \ge 0$, since otherwise  the algorithm reports $|S|$ and exits. The proof of Theorem~\ref{thm:main}, presented in Section~\ref{sec:mainproof}, boils down to establishing the following four propositions:
\begin{enumerate}[(a)]
\item\label{badq}
The probability that $A_q 2^q$ is outside the range 
$(1\pm\delta)|S|$ is at most $\exp\left(-9t/(2b^2)\right)$.
\item\label{badq1}
The probability that $A_{q+1} 2^{q+1}$ is outside the range 
$(1\pm\delta)|S|$ is at most $\exp\left(-9t/(2b^2)\right)$.
\item\label{neverbad}
If $A_q 2^q$ is in the range $(1\pm\delta)|S|$, then the maximum in line~\ref{j_choice} is at least $q$ (deterministically).
\item\label{badq2}
For each $i \ge q+2$, the probability that the maximum in line~\ref{j_choice} equals $i$ is at most $\exp\left(-8t/b^2\right)$.
\end{enumerate}
These propositions imply that the probability of failure is at most the sum of the probability of the bad event in~\ref{badq}, the bad event in~\ref{badq1}, and the (at most) $n-2$ bad events in~\ref{badq2}. The fact that each bad event concerns only one random variable $A_j$ allows a significant acceleration of Algorithm~\ref{alg:simple}, discussed next. 

\section{Nested Sample Sets}

In Algorithm~\ref{alg:simple}, for each $i \in [n]$ and $j \in [t]$, we sample each set $R_{i,j}$ from an $i$-uniform distribution on subsets of $\{0,1\}^n$ independently. As was first pointed out in~\cite{DBLP:conf/ijcai/ChakrabortyMV16}, this is both unnecessary and inefficient. Specifically, as a thought experiment, imagine that we select all random subsets of $\{0,1\}^n$ that we may need \emph{before} Algorithm~\ref{alg:simple} starts, in the following manner (in reality, we only generate the sets as needed).
\begin{algorithm}
\caption{Generates $t$ monotone decreasing sequences of sample sets}\label{alg:nested}
\begin{algorithmic}
\State $R_{0,j} \leftarrow \{0,1\}^n$ for $j \in [t]$ 
\For{$i$ from $1$ to $n$} 
	\For{$j$ from $1$ to $t$}
		\State Select $R_{i,j} \subseteq R_{i-1,j}$ from a $1$-uniform distribution on $R_{i-1,j}$ 
	\EndFor
\EndFor
\end{algorithmic}
\end{algorithm}
 
Organize now these sets in a matrix whose rows correspond to values of $0 \le i \le n$ and whose columns correspond to $j \in [t]$. It is easy to see that:
\begin{enumerate}
\item
For each (row) $i \in [n]$: 
	\begin{enumerate}
	\item\label{same_d}
	Every set $R_{i,j}$ comes from an $i$-uniform distribution on $\{0,1\}^n$. 
	\item\label{still_ind}
	The sets $R_{i,1}, \ldots, R_{i,t}$ are mutually independent. 
	\end{enumerate}
\item
For each column $j \in [t]$:
	\begin{enumerate}
	\item\label{coupled}
	$R_{0,j} \supseteq R_{1,j} \supseteq \cdots \supseteq R_{n-1,j} \supseteq R_{n,j}$. 
	\end{enumerate}
\end{enumerate}

Propositions~\ref{badq}--\ref{badq2} above, in the presence of these new random sets hold exactly as in the fully independent case, since for each fixed $i \in [n]$ the only relevant sets are the sets in row $i$ and their distribution, per (\ref{same_d})--(\ref{still_ind}), does not change. At the same time, (\ref{coupled}) ensures that $Y_{1,j} \ge Y_{2,j} \ge \cdots \ge Y_{n,j}$ for every $j \in [t]$ and, as a result, $Z_1 \ge Z_2 \ge \cdots \ge Z_n$. Therefore, the characteristic function of $A_i = Z_i/t \ge (1-\delta) (4/\delta)$ is now \emph{monotone decreasing}. This means that in order to compute $j$ in line~\ref{j_choice}, instead of computing $Z_i$ for $i$ from $\ell$ to $n$, we can compute $A_{\ell}, A_{\ell+1}, A_{\ell+2}, A_{\ell+4}, A_{\ell+8}, \ldots$ until we encounter our first $k$ such that $A_k <  (1-\delta) (4/\delta)$, say at $k = \ell+2^{c}$, for some $c\ge 0$. At that point, if $c\ge 1$, we can perform binary search for $j \in \{A_{\ell+2^{c-1}}, \ldots, A_{\ell+2^c-1}\}$ etc., so that the number of times the loop that begins in line~\ref{inner_loop} is executed is \emph{logarithmic} instead of linear in $n-\ell$. Moreover, as we will see, the number of iterations $t$ of this loop can now be reduced from $O(\ln(n/\theta))$ to $O(\ln(1/\theta))$, shaving off another $\log n$ factor from the running time.
\section{Proof of Theorem~\ref{thm:main}}\label{sec:mainproof}

To prove Theorem~\ref{thm:main} we will need the following tools.
\begin{lemma}\label{lem:proud}
Let $X \ge 0$ be an arbitrary integer-valued random variable. Write $\ex X = \mu$ and $\mathrm{Var}(X) = \sigma^2$. Let $Y = \min \{X, b\}$, for some integer $b \ge 0$. For any $\lambda > 0$, if $b \ge \mu + \lambda \sigma^2$, then $\ex Y \ge \ex X - 1/\lambda$.
\end{lemma}
\begin{proof}
Recall that if $Z \ge 0$ is an integer-valued random variable, then $\ex Z = \sum_{j>0} \Pr[Z  \ge j]$. Since both $X,Y$ are integer-valued, using Chebychev's inequality to derive~\eqref{eq:cheb}, we see that
\begin{eqnarray}
\ex X - \ex Y & = & \sum_{j > b} \Pr[X \ge j] \nonumber \\
& \le & 
\sum_{t = 1}^{\infty} 
\Pr\left[X \ge \mu +  \lambda \sigma^2 +t \right] \nonumber \\
& = & 
\sum_{t = 1}^{\infty} 
\Pr\left[X \ge \mu +  \sigma\left(\lambda \sigma +\frac{t}{\sigma}\right)\right] \nonumber \\
& \le & 
\sum_{t = 1}^{\infty} \frac{1}{\left(\lambda \sigma+t/\sigma\right)^2} \label{eq:cheb} \\
& \le & 
\int_{t=0}^{\infty} \frac{1}{\left(\lambda \sigma+t/\sigma\right)^2} \, \mathrm{d}t \nonumber \\
& = & \frac{1}{\lambda} \nonumber \enspace .
\end{eqnarray}
\end{proof}

\begin{lemma}\label{lem:cheby}
Let $\mathcal{D}$ be any $i$-uniform distribution on subsets of $\{0,1\}^n$. For any fixed set $S \subseteq \{0,1\}^n$, if $R \sim \mathcal{D}$ and $X = |S \cap R|$, then $\mathrm{Var}(X) \le \ex X + (\mathrm{Boost}({\mathcal D},|S|)-1) (\ex X)^2$.
\end{lemma}
\begin{proof}
Recall that $\mathrm{Var}(X) = \ex X^2 - (\ex X)^2$ and write $\mathbf{1}\{\cdot\}$ for the indicator function. Then
\begin{eqnarray*}\label{eq:core}
\ex X^2 & = & \ex \left(\sum_{\sigma \in S} \mathbf{1}\{\sigma \in R\} \right)^2 \\
& = & \ex\left(\sum_{\sigma,\tau \in S} \mathbf{1}\{\sigma , \tau \in R\}\right) \\
& = & \sum_{\sigma,\tau \in S} \Pr[\sigma , \tau \in R] \\
& = & \sum_{\sigma \in S} \Pr[\sigma \in R] + \sum_{\substack{\sigma, \tau \in S\\ \sigma \neq \tau}} \Pr[\sigma, \tau \in R] \\
& \le & \sum_{\sigma \in S} \Pr[\sigma \in R] + 2^{-2i}|S|(|S|-1) \mathrm{Boost}(\mathcal{D}, |S|) \\
& < & \ex X + \mathrm{Boost}(\mathcal{D},|S|) (\ex X)^2  \enspace .
\end{eqnarray*}
\end{proof}

\begin{proof}[Proof of Theorem~\ref{thm:main}]
Let $q = \lfloor \log_2(\delta|S|/4)\rfloor$. Recall that if $|S| < 4/\delta$, the algorithm returns $|S|$ and exits. Therefore, we can assume without loss of generality that  $q \ge 0$. Fix any $i = q + k$, where $k \ge 0$. Let $X_{i,j} = |S \cap R_{i,j}|$ and write $\ex X_{i,j} = \mu_i$, $\mathrm{Var}(X_{i,j}) = \sigma^2_i$.  

To establish propositions~\ref{badq}, \ref{badq1} observe that the value $\ell$ defined in line~\ref{lin:lll} is at most $q$, since $L \le |S|$, and that $|S| \ge 2^{q+1}$, since $\delta \le 2$. Thus, since $\mathrm{Boost}(\mathcal{D}, M)$ is non-increasing in $M$,
\[
\max_{k \in \{0,1\}} 
\mathrm{Boost}(\mathcal{D}_{q+k}, |S|) 
\le 
\max
\{
\mathrm{Boost}(\mathcal{D}_{q}, 2^{q}), 
\mathrm{Boost}(\mathcal{D}_{q+1}, 2^{q+1}) 
\} \\
 \le  
\max_{\ell \le i \le n}\mathrm{Boost}(\mathcal{D}_i, 2^i) \\
 \le  B  \enspace .
\]
Therefore, we can apply Lemma~\ref{lem:cheby} for $i \in \{q,q+1\}$ and conclude that  $\sigma_i^2 \le \mu_i + (B-1) \mu_i^2$ for such $i$. Since $\mu_i < 8/\delta$ for all $i\ge q$ while $\xi = 8/\delta$, we see that $b =  \lceil\xi + 2( \xi + \xi^2(B-1)) \rceil \ge \mu_i + 2 \sigma_i^2$. Thus, we can conclude that for $i \in \{q,q+1\}$  the random variables $X_{i,j},Y_{i,j}$ satisfy the conditions of Lemma~\ref{lem:proud} with $\lambda=2$, implying $\ex Y_{i,j} \ge \ex X_{i,j} - 1/2$. Since $Z_i$ is the sum of $t$ independent random variables $0 \le Y_{i,j} \le b$ and $\ex Z_i/t
 \ge \mu_i-1/2$, we see that for $i \in \{q,q+1\}$  Hoeffding's inequality implies 
\begin{eqnarray}
\Pr[Z_i/t \le (1-\delta)
\mu_i ] 
& \le & 
 \exp\left(-2t\left(\frac{\delta\mu_i-1/2}{b}\right)^2\right) \label{eq:lower_relative} 
 \enspace .
\end{eqnarray}  

At the same time, since $Z_i$ is the sum of $t$ independent random variables $0 \le Y_{i,j} \le b$ and $\ex Z_i/t \le \mu_i$, we see that for all $i \ge q$, Hoeffding's inequality implies 
\begin{eqnarray}
\Pr[Z_i/t \ge (1+\delta)
\mu_i ]
 & \le & 
 \exp\left(-2t\left(\frac{\delta\mu_i}{b}\right)^2\right) \label{eq:upper_relative}
 \enspace .
\end{eqnarray}

To conclude the proof of propositions~\ref{badq} and~\ref{badq1} observe that $\mu_{q+k} \ge 2^{2-k}/\delta$. Therefore, \eqref{eq:lower_relative} and~\eqref{eq:upper_relative} imply that for $k \in \{0,1\}$, the probability that $A_{q+k} 2^{q+k}$ is outside $(1\pm\delta)|S|$ is at most
\[
2\exp\left(-2t\left(\frac{2^{2-k}-1/2}{b}\right)^2\right) < 2\exp(-9t/(2b^2))\enspace.
\]

To establish proposition~\ref{neverbad} observe that if $A_q \ge (1-\delta) \mu_q$, then $A_q \ge (1-\delta)(4/\delta)$ and, thus, $j \ge q$. Finally, to establish proposition~\ref{badq2} observe that $\mu_i < 2/\delta$ for all $i \ge q+2$. Thus, for any such $i$, requiring $\mu_i + w \ge (1-\delta)(4/\delta)$, implies $w > 2(1-2\delta)/\delta$, which, since $\delta \le 1/3$, implies $w > 2$. Therefore, for every $k\geq 2$, the probability that $j = q+k$ is at most $\exp(-8t/b^2)$.

Having established propositions~\ref{badq}--\ref{badq2} we argue as follows. If $A_{q+k} 2^{q+k}$ is in the range $(1\pm\delta)|S|$ for $k \in \{0,1\}$ and smaller than $(1-\delta)(4/\delta)$ for $k\ge 2$, then the algorithm will report either $A_{q} 2^{q}$ or $A_{q+1} 2^{q+1}$, both of which are in $(1\pm\delta)|S|$. Therefore, the probability that the algorithm's answer is incorrect is at most $2 \cdot 2\exp(-9t/(2b^2)) + n \cdot \exp(-8t/b^2) <\theta$, for $n>2$. 
\end{proof}

\subsection{Proof for Monotone Sequences of Sample Sets}
\begin{theorem}\label{th:mono}
For any $s>0$, if the sets $R_{i,j}$ are generated by Algorithm~\ref{alg:nested} and $t \ge (2b^2/9) \ln(5s)$ in line~\ref{tline}, then the output of Algorithm~\ref{alg:simple} lies in the range $(1\pm \delta)|S|$ with probability at least $1-\exp(-s) > 0$. 
\end{theorem}

\begin{proof}
Observe that for any fixed $i$, since the sets $R_{i,1}, \ldots, R_{i,t}$ are mutually independent, equations~\eqref{eq:lower_relative} and~\eqref{eq:upper_relative} remain valid and, thus, propositions \ref{badq}--\ref{neverbad} hold. For proposition~\ref{badq2} we note that if the inequality $A_{q+k} 2^{q+k} < (1-\delta)(4/\delta)$ holds for $k=2$, then, by monotonicity, it holds for all $k \ge 2$. Thus, all in all, when monotone sequences of sample sets are used, the probability that the algorithm fails is at most 
\[
4\exp(-9t/(2b^2)) + \exp(-8t/b^2) \enspace  ,
\]
a quantity smaller than $\exp(-s)$ for all  $t \ge (2b^2/9) \ln(5s)$.
\end{proof}
\section{Low Density Parity Check Codes}\label{sec:ldpc}

In our earlier discussion of $i$-uniform distributions we saw that if both $A \in \{0,1\}^{i \times n}$ and $b \in \{0,1\}^i$ are uniformly random, then membership in the random set $R = \{\sigma: A\sigma = b\}$ enjoys the (highly desirable) property of pairwise independence. Unfortunately, this also means that each of the $i$ parity constraints involves, on average, $n/2$ variables, making it difficult to work with when $i$ and $n$ are large (we are typically interested in the regime $\log_2 |S(F)| = \Omega(n)$, thus requiring $i=\Omega(n)$ constraints to be added).

The desire to sparsify the matrix $A$ has long been in the purview of the model counting community. Unfortunately, achieving sparsity by letting each entry of $A$ take the value 1 independently with probability $p < 1/2$ is not ideal~\cite{ermon2014low}: the resulting random sets become dramatically ``lumpy" as $p \to 0$. 

A motivation for our work is the realization that if we write $|S| = 2^{\alpha n}$, then as $\alpha$ grows it is possible to choose $A$ so that it is both very sparse, i.e., with each row having a \emph{constant} non-zero elements, and so that the sets $R$ have relatively low lumpiness \emph{at the $2^{\alpha n}$ scale}. The key new ingredient comes from the seminal work of Sipser and Spielman on expander codes~\cite{DBLP:journals/tit/SipserS96} and is this:
\begin{center}
Require each \emph{column} of $A$ to have at least 3 elements . 
\end{center}
Explaining why this very simple modification has profound implications is beyond the scope of this paper. Suffice it to say, that it is precisely this requirement of minimum variable degree that dramatically reduces the correlation between elements of $R$ and thus $\mathrm{Boost}(\mathcal{D})$. 

\newcommand{\ld}{\mathtt{l}}
\newcommand{\rd}{\mathtt{r}}

For simplicity of exposition, we only discuss matrices $A \in \{0,1\}^{i \times n}$ where:
\begin{itemize}
\item
Every column (variable) has exactly $\ld \ge 3$ non-zero elements.
\item
Every row (parity constraint) has exactly $\rd = \ld n/i \in$ non-zero elements.
\end{itemize}
Naturally, the requirement $\ld n/i \in \mathbb N$ does not always hold, in which case some rows have $\lfloor \ld n/i \rfloor$ variables, while the rest have $\lceil \ld n/i \rceil$ variables, so that the average is $\ld n/i$. To simplify discussion we ignore this point in the following. 

Given $n, i,$ and $\ld$ a (bi-regular) Low Density Parity Check (LDPC) code is generated by selecting a uniformly random matrix as above\footnote{Generating such a matrix can be done by selecting a random permutation of $[\ld n]$ and using it to map each of the $\ld n$ non-zeros to equations, $\rd$ non-zeros at a time; when $\ld,\rd \in O(1)$, the variables in each equation will be distinct with probability $\Omega(1)$, so that a handful of trials suffice to generate a matrix as desired.} and taking the set of codewords to be the set of solutions of the linear system $A\sigma = \mathbf{0}$. (While, for model counting we must also take the right hand side of the equation to be a uniformly random vector, when talking about the geometric properties of the set of solutions, due to symmetry we can assume without loss of generality that $b = \mathbf{0}$.) In particular, note that $\sigma = \mathbf{0}$ is always a solution of the system and, therefore, to discuss the remaining solutions (codewords) instead of referring to them by their distance from our reference solution $\sigma = \mathbf{0}$ we can refer to them by their weight, i.e., their number of ones. 

It is well-known~\cite{mct} that the expected number of codewords of weight $w$ in a bi-regular LDPC code is given by the following (rather complicated) expression. 
\begin{lemma}[Average weight-distribution of regular LDPC ensembles]\label{lem:rudy}
The expected number of codewords of weight $w$ in a bi-regular LDPC code with $n$ variables and $i$ parity equations, where each variable appears in $\ld$ equations and each equation includes $\rd$ variables equals the coefficient of $x^{w \ld}$ in the polynomial
\begin{equation}\label{eq:rudy}
\binom{n}{w}
\frac{\left(\sum_{i} \binom{r}{2i} x^{2i}\right)^{n \frac{\ld}{\rd}}}{\binom{n \ld}{w \ld}} \enspace .
\end{equation}
\end{lemma}
We will denote the quantity described in Lemma~\ref{lem:rudy} by $\mathrm{codewords}(w)$. 

\section{Tractable Distributions}\label{sec:moufa}

Let $\mathcal{D}_i$ be any $i$-uniform distribution on subsets of $\{0,1\}^n$.
\begin{definition}
Say that $\mathcal{D}_i$ is \emph{tractable} if there exists a function $f$, called the \emph{density} of $\mathcal{D}_i$, such that for all $\sigma, \tau \in \{0,1\}^n$, if $R \sim \mathcal{D}_i$, then $\Pr[\tau \in R \mid \sigma \in R] = f(\mathrm{Hamming}(\sigma,\tau))$, where
\begin{itemize}
\item
$f(j) \ge f(j+1)$ for all $j < n/2$, and, 
\item
either $f(j) \ge f(j+1)$ for all $j \ge n/2$, or $f(j) = f(n-j)$ for all $j \ge n/2$. 
\end{itemize}
\end{definition}

For any $S \subset \{0,1\}^n$ and $\sigma \in S$, let $H_{\sigma}(d)$ denote the number of elements of $S$ at Hamming distance $d$ from $\sigma$. Recalling the definition of $\mathrm{Boost}$ in~\eqref{eq:recboost}, we get~\eqref{eq:tyu} by $i$-uniformity and~\eqref{eq:kopi} by tractability,
\begin{eqnarray}
\mathrm{Boost}(\mathcal{D}_i, M) & = & 
\max_{\substack{S \subseteq \{0,1\}^n\\ |S| \ge M}}
\frac{1}{|S|(|S|-1)}
\sum_{\substack{\sigma, \tau \in S\\ \sigma \neq \tau}}\frac{\Pr[\sigma, \tau \in R]}{\Pr[\sigma \in R] \Pr[\tau \in R]} \label{eq:recboost} \\
& = & 
\max_{\substack{S \subseteq \{0,1\}^n\\ |S| \ge M}}
\frac{2^{i}}{|S|(|S|-1)}
\sum_{\sigma \in S} \sum_{\tau \in S-\sigma} \Pr[\tau \in S \mid \sigma \in S] \label{eq:tyu}\\
& = & 
\max_{\substack{S \subseteq \{0,1\}^n\\ |S| \ge M}}
\frac{2^{i}}{|S|(|S|-1)}
\sum_{\sigma \in S} \sum_{d=1}^n H_{\sigma}(d) f(d)
\label{eq:kopi} \\
& \le & 
\max_{\substack{S \subseteq \{0,1\}^n\\ |S| \ge M\\ \sigma \in S}}
\frac{2^{i}}{|S|-1}
\sum_{d=1}^n H_{\sigma}(d) f(d) \label{eq:lalakis}
\enspace .
\end{eqnarray}

Let $z$ be the unique integer such that $|S|/2= \binom{n}{0}+\binom{n}{1}+\cdots+\binom{n}{z-1}+\alpha \binom{n}{z}$, for some $\alpha \in [0,1)$. Since $z \le n/2$, tractability implies that $f(j) \ge f(j+1)$ for all $0 \le d < z$, and therefore that
\begin{eqnarray}
\frac{\sum_{d=1}^n H_{\sigma}(d) f(d)}{|S|-1} \label{eq:miniscule}
& \le & 
\frac{\sum_{d=0}^n H_{\sigma}(d) f(d)}{|S|} \\
& \le & 
\frac{\sum_{d=0}^{n/2} H_{\sigma}(d) f(d) + \sum_{d>n/2} H_{\sigma}(d) f(n-d)}{|S|}\\
& \le & 
\frac{2 \left(\sum_{d=0}^{z-1} \binom{n}{d} f(d) + \alpha \binom{n}{z} f(z)\right)}{|S|} \\
& = & 
\frac{\sum_{d=0}^{z-1} \binom{n}{d} f(d) + \alpha \binom{n}{z} f(z)}
{\sum_{d=0}^{z-1} \binom{n}{d} + \alpha \binom{n}{z}} \\
& \le & 
\frac{\sum_{d=0}^{z-1} \binom{n}{d} f(d)}
{\sum_{d=0}^{z-1} \binom{n}{d}}  \\
& := & B(z) \enspace . \label{eq:miha}
\end{eqnarray}

To bound $B(z)$ observe that $B(j) \ge B(j+1)$ for $j < n/2$, inherited by the same property of $f$. Thus, to bound $B(z)$ from above it suffices to bound $z$ for below. Let $h : x \mapsto -x\log_2 x -(1-x) \log_2 x$ be the binary entropy function and let $h^{-1} : [0,1] \mapsto [0,1]$ map $y$ to the smallest number $x$ such that $h(x)=y$. It is well-known that $\sum_{d=0}^z \binom{n}{d} \le 2^{n h(z/n)}$, for every integer $0 \le z \le n/2$. Therefore, $z \ge \lceil n h^{-1}(\log_2(|S|/2)/n)\rceil$, which combined with~\eqref{eq:lalakis} implies the following. 
\begin{theorem}\label{thm:lumpy}
If $\mathcal{D}_i$ is a tractable $i$-uniform distribution with density $f$, then 
\begin{equation}\label{eq:supexp}
\mathrm{Boost}(\mathcal{D}_i, M) \le 
2^{i} B\left(\left\lceil n h^{-1}\left(\frac{\log_2 M - 1}{n}\right)\right\rceil \right) \enspace ,
\end{equation} where 
$B(z) =
{\sum_{d=0}^{z-1} \binom{n}{d} f(d)}/
{\sum_{d=0}^{z-1} \binom{n}{d}}$ 
and $h^{-1} : [0,1] \mapsto [0,1]$ maps $y$ to the smallest number $x$ such that $h(x)=y$, where $h$ is the binary entropy function.
\end{theorem}

Before proceeding to discuss the tractability of LDPC codes, let us observe that the bound in~\eqref{eq:miha} is essentially tight, as demonstrated when $S$ comprises a Hamming ball of radius $z$ centered at $\sigma$ and a Hamming ball of radius $z$ centered at $\sigma$'s complement (in which case the only (and miniscule) compromise is~\eqref{eq:miniscule}). On the other hand, the passage from~\eqref{eq:kopi} to~\eqref{eq:lalakis}, mimicking the analysis of~\cite{ermon2014low}, allows the aforementioned worst case scenario to occur \emph{simultaneously} for every $\sigma \in S$ , an impossibility. As $|S|/2^n$ grows, this is increasingly pessimistic.

\subsection{The Lumpiness of LDPC Codes}
Let $\mathcal{D}_i$ be the $i$-uniform distribution on subsets of $\{0,1\}^n$ that results when $R = \{\sigma: A\sigma =b\}$, where $A$ corresponds to a biregular LDPC code with $i$ parity equations. The row- and column-symmetry in the distribution of $A$ implies that  the function $f(d) = \mathrm{codewords}(d)/\binom{n}{d}$ is the density of $\mathcal{D}_i$. Regarding tractability, it is easy to see that if $\ld$ is odd, then $\mathrm{codewords}(2j+1) = 0$ for all $j$ and that if $\rd$ is even, then $\mathrm{codewords}(d) = \mathrm{codewords}(n-d)$ for all $d$. Thus, for simplicity of exposition, we will restrict to the case where both $\ld$ and $\rd$ are even, noting that this is not a substantial restriction. 

With $\ld,\rd$ even, we are left to establish that $f(j) \ge f(j+1)$ for all $0 \le j < n/2$. Unfortunately, this is not true for a trivial reason: $f$ is non-monotone in the vicinity of $n/2$, exhibiting minisucle finite-scale-effect fluctuations (around its globally minimum value). While this renders Theorem~\ref{thm:lumpy} inapplicable, it is easy to overcome. Morally, because $f$ is \emph{asymptotically} monotone, i.e., for any fixed $\beta \in [0,1/2)$, the inequality $f(\beta n) \ge f(\beta n+1)$ holds for all $n \ge n_0(\beta)$. Practically, because for the proof of Theorem~\ref{thm:lumpy} to go through it is enough that $f(j) \ge f(j+1)$ for all $0 \le j < z$ (instead of all $0 \le j < n/2$), something which for most sets of interest holds, as $z \ll n/2$. Thus, in order to provide a rigorous upper bound on $\mathrm{Boost}$, as required in Algorithm~\ref{alg:simple}, it is enough to verify the monotonicity of $f$ up to $z$ in the course of evaluating $B(z)$. This is precisely what we did with $\ld =8$, $\log_2 M = 2n/5$, and $n \in \{100, 110, \ldots, 200\}$, resulting in $\rd=20$, i.e., equations of length 20. The resulting bounds for $B$ are in Table~\ref{tab:boost} below.
\begin{longtable}{c | c | c | c | c | c | c | c | c | c | c | c }
\caption{Upper bounds for $\mathrm{Boost}$ for equations of length 20.}  \label{tab:boost} \\
$n$ & 100 & 110  & 120 & 130 & 140 & 150 & 160  & 170 & 180 & 190 & 200 \\
\hline
$\mathrm{Boost}$ & 75 & 50  & 35 & 26 & 134 & 89 & 60  & 44 & 34 & 154 & 105 \\
\end{longtable}

Several comments are due here. First, the non-monotonicity of the bound is due to the interaction of several factors in~\eqref{eq:supexp}, most anomalous of which is the ceiling. Second, recall that the running time of Algorithm~\ref{alg:simple} is proportional to the square of our upper bound for $\mathrm{Boost}$. The bounds in Table~\ref{tab:boost} allow us to derive rigorous results for systems with $40-80$ equations and $n \in [100,200]$ after $\sim 10^4$ (parallelizable)  solver invocations. While this is certainly not ideal, any results for such settings are completely outside the reach of CryptoMiniSAT (and, thus, ApproxMC2) when equations of length $n/2$ are used. Finally, as we will see in Section~\ref{sec:experiment}, these bounds on $\mathrm{Boost}$ appear to be \emph{extremely} pessimistic in practice.

\section{Experiments}\label{sec:experiment}

The goal of our this section is to demonstrate the promise of using systems of parity equations corresponding to LDPC codes empirically. That is, we will use such systems, but make \emph{far fewer} solver invocations than what is mandated by our theoretical bounds for a high probability approximation. In other words, {\bf our results are not guaranteed, unlike those of ApproxMC2}. The reason we do this is because we believe that while the error-probability analysis of Theorems~\ref{thm:main} and~\ref{th:mono} is not too far off the mark, the same can not be said for Theorem~\ref{thm:lumpy}, providing our rigorous upper bound on $\mathrm{Boost}$. 

To illuminate the bigger picture, besides ApproxMC2 we also included in the comparison the exact \texttt{sharpSAT} model counter of Thurley~\cite{Thurley}, and the modification of ApproxMC2 in which each equation involves each variable independently with probability $p=1/2^j$, for $j=2,\ldots,5$. (ApproxMC2 uses $j=1$). To make the demonstration as transparent as possible, we only made two modifications to ApproxMC2 and recorded their impact on performance. 
\begin{itemize}
\item
We incorporated Algorithm~\ref{alg:augment} for quickly computing a lower bound. 
\item
We use systems of equations corresponding to LDPC codes instead of systems where each equation involves $n/2$ variables on average (as ApproxMC2 does).
\end{itemize}
Algorithm~\ref{alg:augment} is invoked at most once, while the change in the systems of equations is entirely encapsulated in the part of the code generating the random systems. No other changes to ApproxMC2 (AMC2) were made.

We consider the same 387 formulas as~\cite{DBLP:conf/ijcai/ChakrabortyMV16}. Among these are 2 unsatisfiable formulas which we removed. We also removed 9 formulas that were only solved by \texttt{sharpSAT} and 10 formulas whose number of solutions (and, thus, equations) is so small that the LDPC equations devolve into long XOR equations. Of the remaining 366 formulas, \texttt{sharpSAT} solves 233 in less than 1 second, in every case significantly faster than all sampling based methods. At the other extreme, 46 formulas are not solved by any method within the given time limits, namely 8 hours per method-formula pair (and 50 minutes for each solver invocation for the sampling based algorithms). We report on our experiments with the remaining 87 formulas.  All experiments were run on a modern cluster of 13 nodes, each with 16 cores and 128GB RAM. 

Our findings can be summarized as follows:
\begin{enumerate}
\item
The LDPC-modified version of AMC2 has similar accuracy to AMC2, even though the number of solver invocations is much smaller than what theory mandates for a guarantee. Specifically, the counts are very close to the counts returned by AMC2 and \texttt{sharpSAT}  in \emph{every single} formula.
\item
The counts with $p=1/4$ are as accurate as with $p=1/2$. But for $p \le 1/8$, the counts are often significantly wrong and we don't report results for such $p$.
\item
The LDPC-modified version of AMC2 is faster than AMC2 in \emph{all but one} formulas, the speedup typically exceeding 10x and often exceeding 50x. 
\item
When both \texttt{sharpSAT} and the LDPC-modified version of AMC2 terminate, \texttt{sharpSAT} is faster more often than not. That said, victories by a speed factor of 50x occur for both algorithms.
\item
The LDPC-modified version of AMC2 did not time out on \emph{any} formula. In contrast, \texttt{sharpSAT} timed out on 38\% of the formulas,  AMC2 with $p=1/4$ on 59\% of the formulas, and AMC2 on 62\%.
\end{enumerate}

In the following table, the first four numerical columns report the binary logarithm of the estimate of $|S|$ returned by each algorithm. The next four columns report the time taken to produce the estimate, in seconds. We note that several of the 87 formulas come with a desired \emph{sampling set}, i.e., a subset of variables  $V$ such that the goal is to count the size of the projection of the set of all models on $V$. Since, unlike AMC2, \texttt{sharpSAT} does not provide such constrained counting functionality, to avoid confusion, we do not report a count for \texttt{sharpSAT} for these formulas, writing ``---" instead. Timeouts are reported as  ``NA".



\begin{longtable}{
|>{\tiny}l| c | c | c | c || c | c | c | c |}
\caption{Estimates of $\log_2|\text{\# models}|$, followed by time in seconds.}\label{tab:n} \\
\hline
{\normalsize Formula Name} &  \#SAT & LDPC & AMC2 & 1/4 & \#SAT & LDPC & AMC2 & 1/4 \\
\hline
jburnim\_morton.sk\_13\_530 					& NA	& 248.49 & NA	& NA 	& NA & 27826.4 & NA & NA\\ 
blasted\_case37 								& NA 	& 151.02 & NA 	& NA 	& NA & 4149.9 & NA & NA\\
blasted\_case\_0\_b12\_even1 					& NA 	& 147.02 & NA 	& NA 	& NA & 1378.8 & NA & NA\\
blasted\_case\_2\_b12\_even1 					& NA 	& 147.02 & NA 	& NA 	& NA & 1157.5 & NA & NA\\
blasted\_case42 								& NA 	& 147.02 & NA 	& NA 	& NA & 1008.0 & NA & NA\\
blasted\_case\_1\_b12\_even1 					& NA 	& 147.02 & NA 	& NA 	& NA & 1102.0 & NA & NA\\
blasted\_case\_0\_b12\_even2 					& NA 	& 144.02 & NA 	& NA 	& NA & 881.6 & NA & NA\\
blasted\_case\_1\_b12\_even2 					& NA 	& 144.02 & NA 	& NA 	& NA & 1156.3 & NA & NA\\
blasted\_case\_2\_b12\_even2 					& NA 	& 144.02 & NA 	& NA 	& NA & 1050.5 & NA & NA\\
blasted\_case\_3\_4\_b14\_even 				& NA 	& 138.02 & NA 	& NA 	& NA & 293.4 & NA & NA\\
blasted\_case\_1\_4\_b14\_even 				& NA 	& 138.02 & NA 	& NA 	& NA & 472.6 & NA & NA\\
log2.sk\_72\_391 								& --- 	& 136.00 & NA 	& NA 	& ---  & 12811.1 & NA & NA\\
blasted\_case1\_b14\_even3 					& NA 	& 122.02 & NA 	& NA 	& NA & 169.6 & NA & NA\\
blasted\_case\_2\_b14\_even 					& NA 	& 118.02 & NA 	& NA 	& NA & 89.2 & NA & NA\\
blasted\_case3\_b14\_even3 					& NA 	& 118.02 & NA 	& NA 	& NA & 107.7 & NA & NA\\
blasted\_case\_1\_b14\_even 					& NA 	& 118.02 & NA 	& NA 	& NA & 94.7 & NA & NA\\
partition.sk\_22\_155 							& NA 	& 107.17	& NA 	& NA 	& NA & 5282.3 & NA & NA\\
scenarios\_tree\_delete4.sb.pl.sk\_4\_114 			& --- 	& 105.09	& NA 	& NA 	& --- & 708.4 & NA & NA\\
blasted\_case140 								& NA 	& 103.02	& NA 	& NA 	& NA & 1869.0 & NA & NA\\
scenarios\_tree\_search.sb.pl.sk\_11\_136 			& NA 	& 96.46 	& NA 	& NA 	& NA & 3314.2 & NA & NA\\
s1423a\_7\_4 								& 90.59 	& 90.58 	& NA 	& NA 	& 6.2 & 32.4 & NA & NA\\
s1423a\_3\_2 								& 90.16 	& 90.17 	& NA 	& NA 	& 5.7 & 28.3 & NA & NA\\
s1423a\_15\_7 								& 89.84 	& 89.83 	& NA 	& NA 	& 13.6 & 44.8 & NA & NA\\
scenarios\_tree\_delete1.sb.pl.sk\_3\_114 			& --- 	& 89.15 	& NA 	& NA 	& --- & 431.3 & NA & NA\\
blasted\_case\_0\_ptb\_2 						& NA 	& 88.02 	& NA 	& NA 	& NA & 463.6 & NA & NA\\
blasted\_case\_0\_ptb\_1 						& NA 	& 87.98 	& NA 	& NA 	& NA & 632.0 & NA & NA\\
scenarios\_tree\_delete2.sb.pl.sk\_8\_114 			& --- 	& 86.46 	& NA 	& NA 	& --- & 210.3 & NA & NA\\
scenarios\_aig\_traverse.sb.pl.sk\_5\_102 			& NA 	& 86.39 	& NA 	& NA 	& NA & 3230.0 & NA & NA\\
54.sk\_12\_97 								& 82.50 	& 81.55 	& NA 	& NA 	& 20.4 & 235.8 & NA & NA\\
blasted\_case\_0\_b14\_1 						& 79.00 	& 79.09 	& NA 	& NA 	& 28.8 & 33.5 & NA & NA\\
blasted\_case\_2\_ptb\_1 						& NA 	& 77.02 	& NA 	& NA 	& NA & 10.1 & NA & NA\\
blasted\_case\_1\_ptb\_1 						& NA 	& 77.02 	& NA 	& NA 	& NA & 9.5 & NA & NA\\
blasted\_case\_1\_ptb\_2 						& NA 	& 77.02 	& NA 	& NA 	& NA & 17.8 & NA & NA\\
blasted\_case\_2\_ptb\_2 						& NA 	& 77.00 	& NA 	& NA 	& NA & 25.0 & NA & NA\\
blasted\_squaring70 							& 66.00 	& 66.04 	& NA 	& NA 	& 5822.7 & 87.7 & NA & NA\\
blasted\_case19 								& 66.00 	& 66.02 	& NA 	& NA 	& 25.1 & 6.9 & NA & NA\\
blasted\_case20 								& 66.00 	& 66.02 	& NA 	& NA 	& 2.0 & 4.4 & NA & NA\\
blasted\_case15 								& 65.00 	& 65.02 	& NA 	& NA	& 172.3 & 12.4 & NA & NA\\
blasted\_case10 								& 65.00 	& 65.02 	& NA 	& NA 	& 209.8 & 8.8 & NA & NA\\
blasted\_TR\_b12\_2\_linear 					& NA 	& 63.93 	& NA 	& NA 	& NA & 1867.1 & NA & NA\\
blasted\_case12 								& NA 	& 62.02 	& NA 	& NA 	& NA & 21.5 & NA & NA\\
blasted\_case49 								& 61.00 	& 61.02 	& NA 	& NA 	& 8.9 & 15.6 & NA & NA\\
blasted\_TR\_b12\_1\_linear 					& NA 	& 59.95 	& NA 	& NA 	& NA & 767.9 & NA & NA\\
scenarios\_tree\_insert\_insert.sb.pl.sk\_3\_68 		& --- 	& 51.86 	& NA 	& NA 	& 12.1 & 54.3 & NA & NA\\
blasted\_case18 								& NA 	& 51.00 	& NA 	& NA 	& NA & 16.7 & NA & NA\\
blasted\_case14 								& 49.00 	& 49.07 	& NA 	& NA 	& 117.2 & 7.6 & NA & NA\\
blasted\_case9 								& 49.00 	& 49.02 	& NA 	& NA 	& 123.6 & 7.1 & NA & NA\\
blasted\_case61 								& 48.00 	& 48.02 	& NA 	& NA 	& 154.2 & 6.7 & NA & NA\\
ProjectService3.sk\_12\_55 						& --- 	& 46.55 	& 46.58 	& 46.55 	& --- & 12.9 & 273.4 & 267.1\\
blasted\_case145 								& 46.00 	& 46.02 	& NA 	& 46.02 	& 29.2 & 8.4 & NA & 5570.4\\
blasted\_case146 								& 46.00 	& 46.02 	& 46.02 	& NA 	& 29.3 & 4.8 & 9528.6 & NA\\
ProcessBean.sk\_8\_64 							& --- 	& 42.83 	& 42.91 	& 42.83 	& --- & 17.0 & 323.2 & 207.3\\
blasted\_case106 								& 42.00 	& 42.02 	& 42.02 	& 42.02 	& 10.2 & 3.3 & 325.0 & 14728.3\\
blasted\_case105 								& 41.00 	& 41.00 	& 41.04 	& NA 	& 7.5 & 4.0 & 368.5 & NA\\
blasted\_squaring16 							& 40.76 	& 40.83 	& NA	& 41.07 	& 99.4 & 50.3 & NA & 1633.3\\
blasted\_squaring14 							& 40.76 	& 40.70 	& NA 	& 41.00 	& 102.1 & 34.3 & NA & 2926.5\\
blasted\_squaring12 							& 40.76 	& 40.61 	& NA 	& 41.00 	& 117.3 & 39.6 & NA & 1315.6\\
blasted\_squaring7 							& 38.00 	& 38.29 	& 38.00 	& 38.11 	& 45.4 & 34.9 & 432.4 & 263.2\\
blasted\_squaring9 							& 38.00 	& 38.04 	& 37.98 	& 38.15 	& 36.3 & 24.2 & 489.8 & 238.6\\
blasted\_case\_2\_b12\_2 						& 38.00 	& 38.02 	& 38.02 	& 38.00 	& 29.3 & 4.4 & 186.8 & 87.2\\
blasted\_case\_0\_b11\_1 						& 38.00 	& 38.02 	& 38.02 	& 38.04 	& 45.5 & 2.5 & 190.4 & 180.7\\
blasted\_case\_0\_b12\_2 						& 38.00 	& 38.02 	& 38.02 	& 38.02 	& 29.2 & 3.8 & 181.1 & 69.9\\
blasted\_case\_1\_b11\_1 						& 38.00 	& 38.02 	& 38.02 	& 37.81 	& 45.2 & 3.5 & 159.5 & 119.2\\
blasted\_case\_1\_b12\_2 						& 38.00 	& 38.02 	& 38.02 	& 38.02 	& 30.6 & 2.9 & 185.3 & 80.0\\
blasted\_squaring10 							& 38.00 	& 38.02 	& 37.91 	& 38.04 	& 17.6 & 32.0 & 415.1 & 221.7\\
blasted\_squaring11 							& 38.00 	& 37.95 	& 38.02 	& 38.09 	& 19.8 & 19.7 & 470.1 & 207.3\\
blasted\_squaring8 							& 38.00 	& 37.93 	& 38.09 	& 39.00 	& 18.6 & 28.0 & 431.5 & 727.8\\
sort.sk\_8\_52 								& --- 	& 36.43 	& 36.43 	& 36.36 	& --- & 92.0 & 339.2 & 156.8\\
blasted\_squaring1 							& 36.00 	& 36.07 	& 36.07 	& 36.00 	& 6.6 & 20.0 & 367.8 & 156.9\\
blasted\_squaring6 							& 36.00 	& 36.04 	& 36.00 	& 35.93 	& 8.5 & 17.1 & 429.1 & 170.5\\
blasted\_squaring3 							& 36.00 	& 36.02 	& 36.02 	& 36.02 	& 7.7 & 18.7 & 397.3 & 198.5\\
blasted\_squaring5 							& 36.00 	& 35.98 	& 36.02 	& 36.04 	& 8.5 & 28.8 & 384.0 & 228.2\\
blasted\_squaring2 							& 36.00 	& 35.98 	& 36.00 	& 36.07 	& 7.5 & 30.6 & 411.5 & 195.8\\
blasted\_squaring4 							& 36.00 	& 35.95 	& 36.04 	& 35.98 	& 7.9 & 23.2 & 469.8 & 180.0\\
compress.sk\_17\_291 							& NA 	& 34.00 	& NA 	& NA 	& NA & 1898.2 & NA & NA\\
listReverse.sk\_11\_43 							& NA 	& 32.00 	& 32.00 	& 32.00 	& NA & 2995.3 & 2995.3 & 2995.7\\
enqueueSeqSK.sk\_10\_42 						& NA 	& 31.49 	& 31.39 	& 31.43 	& NA & 67.6 & 252.0 & 124.6\\
blasted\_squaring29 							& 26.25 	& 26.36 	& 26.29 	& 26.39 	& 1.3 & 42.7 & 218.7 & 75.2\\
blasted\_squaring28 							& 26.25 	& 26.32 	& 26.36 	& 26.36 	& 1.9 & 57.6 & 185.1 & 59.0\\
blasted\_squaring30 							& 26.25 	& 26.25 	& 26.29 	& 26.17 	& 1.6 & 40.9 & 179.8 & 60.8\\
tutorial3.sk\_4\_31 							& NA 	& 25.29 	& 25.32 	& 25.25 	& NA & 3480.5 & 19658.2 & 2414.7\\
blasted\_squaring51 							& 24.00 	& 24.11 	& 24.15 	& 24.07 	& 1.6 & 4.8 & 49.3 & 5.3\\
blasted\_squaring50 							& 24.00 	& 23.86 	& 24.00 	& 24.02 	& 1.3 & 4.7 & 54.2 & 5.1\\
NotificationServiceImpl2.sk\_10\_36 				& --- 	& 22.64 	& 22.49 	& 22.55 	& --- & 13.7 & 29.6 & 9.6\\
karatsuba.sk\_7\_41 							& --- 	& 20.36 	& NA 	& 20.52 	& --- & 24963.0 & NA & 19899.0\\
LoginService.sk\_20\_34 						& --- 	& 19.49 	& 19.39 	& 19.43 	& --- & 28.1 & 33.0 & 20.7\\
LoginService2.sk\_23\_36 						& --- 	& 17.55 	& 17.43 	& 17.43 	& --- & 72.9 & 40.8 & 32.6\\
\hline
\end{longtable}

\section*{Acknowledgements}
We are grateful to Kuldeep Meel and Moshe Vardi for sharing their code and formulas and for several valuable conversations. We also thank Kostas Zambetakis for several comments on earlier versions.

\bibliographystyle{plain}
\bibliography{Biblio/thinning,Biblio/meel,Biblio/xorcount}

\newpage

\listoflabels

\end{document}